\documentclass[a4paper,10pt]{article}
\usepackage{charter,eulervm}

\usepackage{amsthm,amsmath,amssymb,upgreek,marvosym}
\usepackage{makeidx}  
\usepackage{paralist}
\usepackage{subfig}
\usepackage{tabularx}
\usepackage[usenames,dvipsnames]{color}
\usepackage[pdftex,breaklinks,colorlinks,
    citecolor={blue}, linkcolor={blue},urlcolor=Maroon]{hyperref}
\usepackage{tkz-graph}
\usetikzlibrary{arrows,shapes,decorations.pathmorphing}

\theoremstyle{plain} 
\newtheorem{theorem}{Theorem}[section]
\newtheorem{lemma}[theorem]{Lemma}

\newtheorem{proposition}[theorem]{Proposition}

\def\boxit#1{\vbox{\hrule\hbox{\vrule\kern4pt
  \vbox{\kern1pt#1\kern1pt}
\kern2pt\vrule}\hrule}}

\newcommand{\keywords}[1]{\bigskip \par\noindent
{\small{\em Keywords\/}: #1}}
\newcommand{\opt}[2]{\ensuremath{{\mathtt{#1}(#2)}}}
\newcommand{\comment}[1]{\`$\setminus\!\!\setminus${\em #1}}

\begin{document}

\title{Approximate Association via Dissociation\thanks{ Supported in
    part by the National Natural Science Foundation of China (NSFC)
    under grants 61232001, 61472449, 61572414, and 61420106009, and
    the Hong Kong Research Grants Council (RGC) under grant PolyU
    252026/15E.}}

\author{Jie You\thanks{School of Information Science and Engineering,
    Central South University, Changsha, China.} \thanks{Department of
    Computing, Hong Kong Polytechnic University, Hong Kong, China.}
  \and
  \addtocounter{footnote}{-2} Jianxin Wang\footnotemark
  \and
  Yixin Cao\footnotemark
}

\date{}
\maketitle

\begin{abstract}
  A vertex set $X$ of a graph $G$ is an association set if each component of $G - X$ is a clique, or a dissociation set if each component of $G - X$ is a single vertex or a single edge. Interestingly, $G - X$ is then precisely a graph containing no induced $P_3$'s or containing no $P_3$'s, respectively.  We observe some special structures and show that if none of them exists, then the minimum association set problem can be reduced to the minimum (weighted) dissociation set problem.  This yields the first nontrivial approximation algorithm for association set, and its approximation ratio is 2.5, matching the best result of the closely related cluster editing problem.  The reduction is based on a combinatorial study of modular decomposition of graphs free of these special structures.  Further, a novel algorithmic use of modular decomposition enables us to implement this approach in $O(m n + n^2)$ time.

\end{abstract}

\section{Introduction}
A {\em cluster graph} comprises a family of disjoint cliques, each an
association.  Cluster graphs have been an important model in the
study of clustering objects based on their pairwise similarities,
particularly in computational biology and machine learning
\cite{ben-dor-99-clustering-gene}.  If we represent each object with a
vertex, and add an edge between two objects that are similar, we would
expect a cluster graph.  If this fails, a natural problem is then to
find and exclude a minimum number of vertices such that the rest forms
a cluster graph; this is the {\em association set} problem.  This
problem has recently received significant interest from the community
of parameterized computation, where it is more commonly called {\em
  cluster vertex deletion}
\cite{huffner-10-cluster-vertex-deletion,boral-14-cluster-vertex-deletion}.
The cardinality of a minimum association set of a graph is also known
as its {\em distance to clusters}.  It is one of the few structural
parameters for dense graphs
\cite{doucha-12-cluster-vertex-deletion-as-parameter,chopin-14-target-set-selection},
in contrast with a multitude of structural parameters for sparse
graphs, thereby providing another motivation for this line of
research.  For example, Bruhn et al.~\cite{bruhn-14-boxicity} recently
showed that the boxicity problem (of deciding the minimum $d$ such
that a graph $G$ can be represented as an intersection graph of
axis-aligned boxes in the $d$-dimension Euclidean space) is
fixed-parameter tractable in the vertex cover number.

The association set problem belongs to the family of vertex deletion
problems studied by
Yannakakis~\cite{lewis-80-node-deletion-np,lund-93-approximation-maximum-subgraph}.
The purpose of these problems is to delete the minimum number of
vertices from a graph so that the remaining subgraph satisfies a
hereditary property (i.e., a graph property closed under taking
induced subgraphs) \cite{lewis-80-node-deletion-np}.  It is known that
a hereditary property can be characterized by a (possibly infinite)
set of forbidden induced subgraphs.  In our case, the property is
``being a cluster graph,'' and the forbidden induced subgraphs are
$P_3$'s (i.e., paths on three vertices).  A trivial approximation
algorithm of ratio $3$ can be derived as follows.  We search for
induced $P_3$'s, and if find one, delete all the three vertices.  This
trivial upper bound is hitherto the best known.  Indeed, this is a
simple application of Lund and Yannakakis's
observation~\cite{lund-93-approximation-maximum-subgraph}, which
applies to all graph classes with finite forbidden induced subgraphs.

Closely related is the cluster editing problem, which allows us to
use, instead of vertex deletions, both edge additions and deletions
\cite{ben-dor-99-clustering-gene}.  Approximation algorithms of the
cluster editing problem have been intensively studied, and the current
best approximation ratio is 2.5
\cite{bansal-04-correlation-clustering,charikar-04-approximate-clustering,ailon-08-aggregating-inconsistent-information}.
Our main result is the first nontrivial approximation algorithm for
the association set problem, with a ratio matching the best ratio of
the closely related cluster editing problem.  As usual, $n$ and $m$
denote the numbers of vertices and edges respectively in the input
graph.  Without loss of generality, we assume throughout the paper
that the input graph contains no isolated vertices (vertices of degree
0), hence $n = O(m)$.
\begin{theorem}\label{thm:alg-approx}
  There is an $O(m n)$-time approximation algorithm of ratio~$2.5$ for
  the association set problem.
\end{theorem}

Our approach is to reduce the association set problem to the weighted
dissociation set problem.  Given a vertex-weighted graph, the
\emph{weighted dissociation set} problem asks for a set of vertices
with the minimum weight such that its deletion breaks all $P_3$'s,
thereby leaving a graph of maximum degree $1$ or $0$.  This problem
was first studied by
Yannakakis~\cite{yannakakis-81-bipartite-node-deletion}, who proved
that its unweighted version is already NP-hard on bipartite graphs.
Note that a $P_3$ that is not induced must be in a triangle.  Thus, in
triangle-free graphs, the weighted version of the association set
problem is equivalent to the weighted dissociation set problem. It is
easy to observe that for the association set problem, vertices in a
twin class (i.e., whose vertices have the same closed neighborhood)
are either fully contained in or disjoint from a minimum solution.
This observation inspires us to transform the input graph $G$ into a
vertex-weighted graph $Q$ by identifying each twin class of $G$ with a
vertex of $Q$ whose weight is the size of the corresponding twin
class.  We further observe that there are five small graphs such that
if $G$ has none of them as an induced subgraph, then $Q$ either has a
simple structure, hence trivially solvable, or is triangle-free, and
can be solved using the ratio-2 approximation algorithm for the
weighted dissociation set problem
\cite{tu-11-primal-dual-vcp3,tu-11-k-path-vertex-cover}.  From the
obtained solution for $Q$ we can easily retrieve a solution for the
original graph $G$.  Since each of these five graphs has at most five
vertices and at least two of them need to be deleted to make it
$P_3$-free, the approximation ratio 2.5 follows.

The main idea of this paper appears in the argument justifying the
reduction from the (unweighted) association set problem to the
weighted dissociation set problem.  Indeed, we are able to provide a
stronger algorithmic result that implies the aforementioned
combinatorial result.  We develop an efficient algorithm that detects
one of the five graphs in $G$, solves the problem completely, or
determines that $Q$ is already triangle-free.  Our principal tool is
modular decomposition.  A similar use of modular decomposition was
recently invented by the authors~\cite{cao-15-edge-deletion} in
parameterized algorithms.  It is worth noting that the basic
observation on vertex deletion problems to graph properties with
finite forbidden induced subgraphs has been used on both approximation
and parameterized algorithms, by Lund and
Yannakakis~\cite{lund-93-approximation-maximum-subgraph} and by
Cai~\cite{cai-96-hereditary-graph-modification} respectively.

As a final remark, cluster editing has a $2 k$-vertex kernel
\cite{cao-12-kernel-cluster-editing}, while it remains an open problem
to find a linear-vertex kernel for the association set (cluster vertex
deletion) problem.

\section{Preliminaries}
This paper will be only concerned with undirected and simple graphs.
The vertex set and edge set of a graph $G$ are denoted by $V(G)$ and
$E(G)$ respectively.  For $\ell\ge 3$, let $P_\ell$ and $C_\ell$
denote respectively an induced path and an induced cycle on $\ell$
vertices.  A $C_3$ is also called a \emph{triangle}.  For a given set
$\cal F$ of graphs, a graph $G$ is {\em $\cal F$-free} if it contains
no graph in $\cal F$ as an induced subgraph.  When $\cal F$ consists
of a single graph $F$, we use also {\em $F$-free} for short.  For each
vertex $v$ in $V(G)$, its \emph{neighborhood} and \emph{closed
  neighborhood} are denoted by $N_G(v)$ and $N_G[v]$ respectively.

A subset $M$ of vertices forms a \emph{module} of $G$ if all vertices
in $M$ have the same neighborhood outside $M$.  In other words, for
every pair of vertices $u,v \in M$, a vertex $x \not\in M$ is adjacent
to $u$ if and only if it is adjacent to $v$ as well.  The set $V(G)$
and all singleton vertex sets are modules, called \emph{trivial}.  A
graph on at least four vertices is \emph{prime} if it contains only
trivial modules, e.g., a $P_4$ and a $C_5$.  Given any partition
$\{M_1,\dots,M_p\}$ of $V(G)$ such that $M_i$ for every $1\le i\le p$
is a module of $G$, we can derive a $p$-vertex \emph{quotient graph}
$Q$ such that for any pair of distinct $i, j$ with $1\le i,j\le p$,
the $i$th and $j$th vertices of $Q$ are adjacent if and only if $M_i$
and $M_j$ are adjacent in $G$ (every vertex in $M_i$ is adjacent to
every vertex in $M_j$).  One should be noted that a single-vertex
graph and $G$ itself are both trivial quotient graphs of $G$, defined
by the trivial module partitions $\{V(G)\}$ and $\{\{v_1\}, \ldots,
\{v_n\}\}$ respectively.

A module $M$ is \emph{strong} if for every other module $M'$ that
intersects $M$, one of $M$ and $M'$ is a proper subset of the other.
All trivial modules are clearly strong.  We say that a strong module
$M$ different from $V(G)$ is \emph{maximal} if the only strong module
properly containing $M$ is $V(G)$.  (It can be contained by non-strong
modules, e.g., in a graph that is a clique, the maximal strong modules
are simply the singletons, while every subset of vertices is a
module.) The set of maximal strong modules of $G$ partitions $V(G)$,
and defines a special {quotient graph} of $G$, denoted by $\widetilde
Q(G)$.\footnote{If $G$ is a clique or independent set, then
  $\widetilde Q(G)$ is isomorphic to $G$ and is the largest quotient
  graph of $G$; if $\widetilde Q(G)$ is prime, then it is the smallest
  \emph{nontrivial} quotient graph of $G$, both cardinality-wise and
  inclusion-wise (see Lemma~\ref{lem:prime-quotient}).  Otherwise,
  there can be other quotient graph larger or smaller than $\widetilde
  Q(G)$.}  The reader who is unfamiliar with modular decomposition is
referred to the survey of Habib and Paul \cite{habib-10-survey-md} for
more information.  The following proposition will be crucial for our
algorithm.
\begin{proposition}\cite{gallai-67-transitive-orientation,
    sumner-73-indecomposable}
  \label{lem:quotient-graphs}
  If a graph $G$ is connected, then $\widetilde Q(G)$ is either a
  clique or prime.  Any prime graph contains an induced $P_4$.
\end{proposition}

Let $Q$ be a quotient graph of $G$, and let $M$ be a module of $G$ in
the module partition defining $Q$.  By abuse of notation, we will also
use $M$ to denote the corresponding node of $Q$; hence $M\in V(Q)$ and
$M\subseteq V(G)$, and its meaning will be clear from context.
Accordingly, by $N_G(M)$ we mean those vertices of $G$ adjacent to $M$
in $G$, and by $N_Q(M)$ we mean those nodes of $Q$ adjacent to $M$ in
$Q$---note that the union of those vertices of $G$ represented by
$N_Q(M)$ is exactly $N_G(M)$.  Sets $N_G[M]$ and $N_Q[M]$ are
understood analogously.

The weighted versions of the associated set problem and the
dissociation set problem are formally defined as follows.  
\begin{figure}[h]
  \centering
\fbox{\parbox{0.93\linewidth}{

    Associated set
    \\[1mm]
\begin{tabularx}{\linewidth}{rX}
  \textit{Input:} & A vertex-weighted graph $G$.
  \\
  \textit{Task:} & find a subset $X\subset V(G)$ of the minimum weight
  such that every component of $G - X$ is a clique.
\end{tabularx}
}}
\bigskip

\fbox{\parbox{0.93\linewidth}{

Dissociation set
    \\[1mm]
\begin{tabularx}{\linewidth}{rX}
  \textit{Input:} & A vertex-weighted graph $G$.
  \\
  \textit{Task:} & find a subset $X\subset V(G)$ of the minimum weight
  such that  every component of $G - X$ is a single vertex or a single edge.
\end{tabularx}
}}
\end{figure}

Let \opt{asso}{G} and \opt{diss}{G} denote respectively the weights of
minimum association sets and minimum dissociation sets of a weighted
graph $G$.  It is routine to verify that $\opt{asso}{G} \le
\opt{diss}{G}$.  Their gap can be arbitrarily large, e.g., if $G$ is a
clique on $n$ vertices, then $\opt{asso}{G} = 0$ and $\opt{diss}{G} =
n - 2$.  A vertex set $X$ is an association set or a dissociation set
of a graph $G$ if and only if $G - X$ contains no $P_3$ as an induced
subgraph or as a subgraph, respectively.  The following proposition
follows from the fact that every $P_3$ in a $C_3$-free graph is
induced.
\begin{proposition}\label{lem:triangle-free}
  If a graph $G$ is $C_3$-free, then $\opt{asso}{G} = \opt{diss}{G}$.
\end{proposition}

\begin{theorem}[\cite{tu-11-k-path-vertex-cover,tu-11-primal-dual-vcp3}]
  \label{thm:alg-dissociation}
  There is an $O(m n)$-time approximation algorithm of ratio~$2$ for
  the weighted dissociation set problem.
\end{theorem}

Note that an unweighted graph can be treated as a special weighted
graph where every vertex receives a unit weight.  In this case,
\opt{asso}{G} is the same as the cardinality of the minimum
association set of $G$.

A $\{C_4,P_4\}$-free graph is called a {\em trivially perfect graph}.
A vertex is {\em universal} if it is adjacent to all other vertices in
this graph, i.e., has degree $n -1$. 
\begin{proposition}[\cite{wolk-62, yan-96-trivially-perfect,
    heggernes-07-certifying-fis}]
  \label{lem:trivially-perfect}
  Every connected trivially perfect graph has a {universal vertex}.
  One can in $O(m)$-time either decide that a graph is a trivially
  perfect graph, or detect an induced $P_4$ or $C_4$.
\end{proposition}

\section{The approximation algorithm}
The association set problem admits a naive $3$-approximation algorithm
\cite{lund-93-approximation-maximum-subgraph}.  It finds an induced
$P_3$ and deletes from $G$ all the three vertices in this $P_3$, and
repeats.  Since any minimum association set has to contain some of the
three vertices, the approximation ratio is at most $3$.  A $P_3$ can
be found in linear time, while the process can be repeated at most
$n/3$ times, and thus the algorithm can be implemented in time $O(m
n)$.  We present here a very simple $2.5$-approximation algorithm,
which runs in a high-order polynomial time, and we will show in the
next section how to implement it in an efficient way to achieve the
running time claimed in Theorem~\ref{thm:alg-approx}.

\begin{figure}[h]
  \centering      \footnotesize
  \subfloat[{$C_4$}]{\label{fig:c4}
    \begin{tikzpicture}[auto=left,every node/.style={fill=blue,circle,blue,draw,inner sep=1pt}, every path/.style={thick},scale=.7]
      \node (a) at (-1,0) {};
      \node (c) at (1,0) {};
      \node (b) at (-1,2) {};
      \node (d) at (1,2) {};
      \draw (a) -- (b) -- (d) -- (c) -- (a);
      \node[white] at (1.25,0) {};
      \node[white] at (-1.25,0) {};
    \end{tikzpicture}
  }
  \qquad
  \subfloat[bull]{\label{fig:bull}
    \begin{tikzpicture}[auto=left,every node/.style={fill=blue,circle,blue,draw,inner sep=1pt}, every path/.style={thick},scale=.8]
      
      \node (a) at (-1,1) {};
      \node (c) at (1,1) {};
      \node (b) at (0,0) {};
      \node (d) at (-1,2) {};
      \node (e) at (1,2) {};
      \draw (a) -- (b) -- (c) -- (a);
      \draw (d) -- (a)  (c) -- (e);
    \end{tikzpicture}
  }
  \qquad
  \subfloat[dart]{\label{fig:dart}
    \begin{tikzpicture}[auto=left,every node/.style={fill=blue,circle,blue,draw,inner sep=1pt}, every path/.style={thick},scale=.8]
      \node (a) at (-1,1) {};
      \node (c) at (1,1) {};
      \node (b) at (0,0) {};
      \node (d) at (-1,2) {};
      \node (e) at (1,2) {};
      \draw (a) -- (b) -- (c) -- (a);
      \draw (d) -- (a) -- (e) -- (c);
    \end{tikzpicture}
  }
  \qquad
  \subfloat[fox]{\label{fig:fox}
    \begin{tikzpicture}[auto=left,every node/.style={fill=blue,circle,blue,draw,inner sep=1pt}, every path/.style={thick},scale=.8]
      
      \node (a) at (-1,1) {};
      \node (c) at (1,1) {};
      \node (b) at (0,0) {};
      \node (d) at (-1,2) {};
      \node (e) at (1,2) {};
      \draw (a) -- (b) -- (c) -- (a);
    \draw (d) -- (a) -- (e) -- (c) -- (d);
  \end{tikzpicture}
  }
  \qquad
  \subfloat[gem]{\label{fig:triangle}
    \begin{tikzpicture}[auto=left,every node/.style={fill=blue,circle,blue,draw,inner sep=1pt}, every path/.style={thick},scale=.8]
      \node (a) at (-1,0.5) {};
      \node (b) at (-.5,0) {};
      \node (c) at (.5,0) {};
      \node (d) at (1,0.5) {};
      \node (e) at (0,2) {};
      \draw (e) -- (a) -- (b) -- (c) -- (d) -- (e);
      \draw (b) -- (e) -- (c);
    \end{tikzpicture}
  }
  \caption{Small subgraphs on 4 or~5 vertices.}
  \label{fig:small-graphs}
\end{figure}
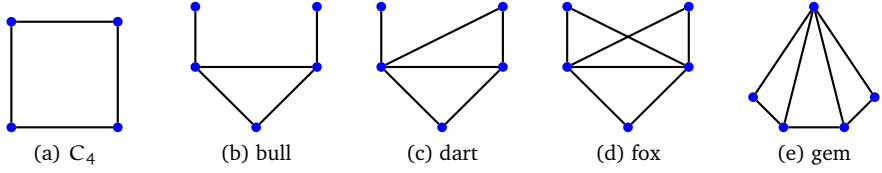

Let ${\cal F}$ denote the set of five small graphs depicted in
Fig.~\ref{fig:small-graphs}, i.e., \{$C_4$, bull, dart, fox, gem\}.  A
quick glance of Fig.~\ref{fig:small-graphs} convinces us that from
each induced subgraph in $\cal F$, at least two vertices need to be
deleted to make it $P_3$-free.
\begin{proposition}\label{lem:small-subgraph}
  Let $X\subseteq V(G)$.  If $G[X]\in \cal F$, then $\opt{asso}{G - X}
  \le \opt{asso}{G} - 2$.
\end{proposition}

In polynomial time we can decide whether $G$ contains an induced
subgraph in $\cal F$, and find one if it exists.  We delete all its
vertices if it is found.  If $G$ is not connected, then we work on its
components one by one.  In the rest of this section we may focus on
connected $\cal F$-free graphs.  In such a graph, every nontrivial
module $M$ induces a $\{C_4,P_4\}$-free subgraph: A $P_4$ in $G[M]$,
together with any $v\in N_G(M)$ (it exists because $G$ is connected
and $M$ is nontrivial), makes a gem.  

One may use the definition of modular decomposition to derive the
following combinatorial properties of $\cal F$-free graphs.  But as we
will present a stronger result in the next section that implies this
lemma, its proof is omitted here.

\begin{lemma}\label{lem:quotient-properties}
  Let $G$ be an $\cal F$-free graph that is not a clique, and let $Q =
  \widetilde Q(G)$.  Either $G$ consists of a set of universal
  vertices and two disjoint cliques, or $Q$ is $C_3$-free and the
  following hold for every maximal strong module $M$ of $G$:
  \begin{enumerate}[(1)]
  \item The subgraph $G[M]$ is a cluster graph.  If it is not a
    clique, then $|N_G(M)| = 1$.
  \item If $|N_Q(M)| > 2$, then the module $M$ is trivial (consisting of
    a single vertex of $G$).
  \end{enumerate}
\end{lemma}
In the first case, $G$ has simply two intersecting cliques $C_1$ and
$C_2$, and the problem is trivial: We delete either $C_1\cap C_2$
(i.e., all universal vertices), or one of $C_1\setminus C_2$ and
$C_1\setminus C_2$, whichever the smaller.  Therefore, we focus on the
other case where $\widetilde Q(G)$ is $C_3$-free.  On the other hand,
if some maximal strong module $M$ does not induce a clique in a
connected $\cal F$-free graph $G$, then we can delete the unique
neighbor of $M$ and consider the smaller graph $G - N_G[M]$.  Now that
$G$ is not a clique but every maximal strong module $M$ of $G$ is, we
can define a vertex-weighted graph $Q$ isomorphic to the quotient
graph $\widetilde Q(G)$, where the weight of each vertex in $Q$ is the
number of vertices in the corresponding module, i.e., $|M|$.  We apply
the algorithm of Tu and Zhou~\cite{tu-11-k-path-vertex-cover} to find
a dissociation set of this weighted graph $Q$.  Since $Q$ is
$C_3$-free, by Proposition~\ref{lem:triangle-free} and
Theorem~\ref{thm:alg-dissociation}, the total weight of the obtained
dissociation set is at most $2\opt{diss}{Q} = 2\opt{asso}{Q} =
2\opt{asso}{G}$.  Putting together these steps, an approximation
algorithm with ratio $2.5$ follows (see
Fig.~\ref{fig:approximation-outline}).

\begin{figure}[h!]
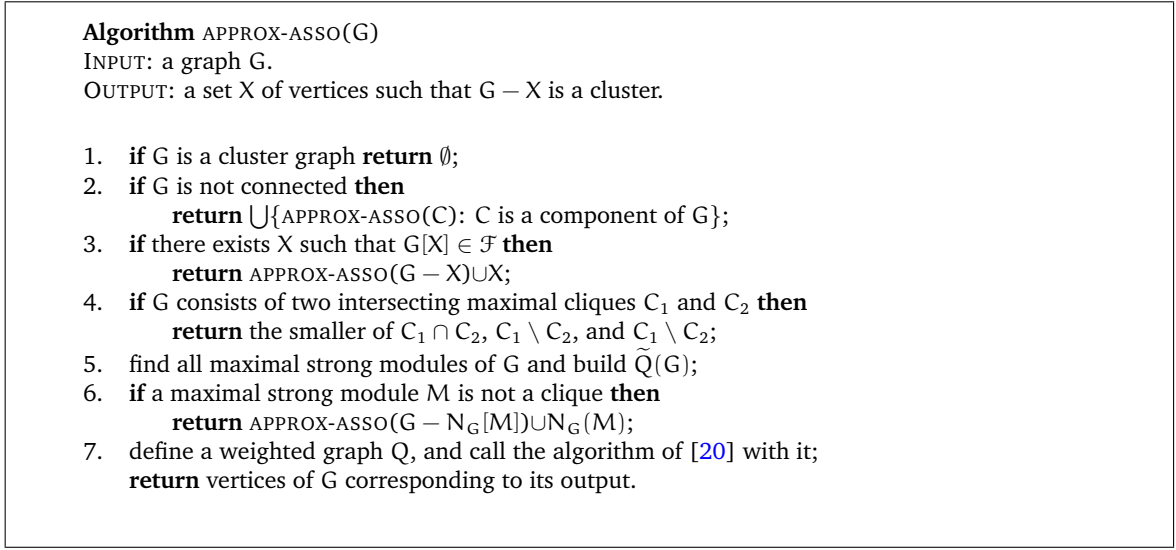

  \vspace*{-5mm}
  \setbox4=\vbox{\hsize28pc \noindent\strut
  \begin{quote}
  \vspace*{-5mm} \small

  {\bf Algorithm \sc approx-asso}($G$)
  \\
  {\sc Input}: a graph $G$.
  \\
  {\sc Output}: a set $X$ of vertices such that $G - X$ is a cluster.

  \begin{tabbing}
    AAA\=AAa\=AAA\=AAA\=AAAAAAAAAAAAAAAAAAAAAAAAAAAA\=A \kill
    1.\> {\bf if} $G$ is a cluster graph {\bf return} $\emptyset$;
    \\
    2.\> {\bf if} $G$ is not connected {\bf then} 
    \\
    \>\>{\bf return} 
    $\bigcup$\{{\sc approx-asso}($C$): $C$ is a component of $G$\};
    \\
    3.\> {\bf if} there exists $X$ such that  $G[X]\in \cal F$ {\bf then}
    \\
    \>\>{\bf return} {\sc approx-asso}($G - X$)$\cup X$;
    \\
    4.\> {\bf if} $G$ consists of two intersecting maximal cliques $C_1$ and $C_2$ {\bf then} 
    \\
    \>\> {\bf return} the smaller of $C_1\cap C_2$, $C_1\setminus C_2$, and
    $C_1\setminus C_2$;
    \\
    5.\> find all  maximal strong modules of $G$ and build $\widetilde Q(G)$;
    \\
    6.\> {\bf if} a maximal strong module $M$ is not a clique {\bf then} 
    \\
    \>\>{\bf return} {\sc approx-asso}($G - N_G[M]$)$\cup N_G(M)$;
    \\
    7.\> define a weighted graph $Q$, and call the algorithm of 
    \cite{tu-11-k-path-vertex-cover} with it;   
    \\
    \> {\bf return} vertices of $G$ corresponding to its output.

  \end{tabbing}
\end{quote} \vspace*{-3mm} \strut} $$\boxit{\box4}$$
\vspace*{-7mm}
\caption{Outline of the approximation algorithm for association set.}
\label{fig:approximation-outline}
\end{figure}

\begin{theorem}
  The output of algorithm {\sc approx-asso}($G$) is an association set
  of the input graph $G$ and its size is at most $2.5 \opt{asso}{G}$.
\end{theorem}
\begin{proof}
  We show that the output is an association set of $G$, and for each
  recursive call that puts $X$ into the solution, it holds
  \[
  |X| \le 2.5 (\opt{asso}{G} - \opt{asso}{G - X}).
  \]
  Steps~1 and~2 are straightforward.  Step~3 follows from
  Proposition~\ref{lem:small-subgraph}, and step~4 is easy to verify.
  Now that $G$ is $\cal F$-free, the second case of
  Lemma~\ref{lem:quotient-properties} applies to $\widetilde Q(G)$
  built in step~5.  In particular, the situations dealt by step~6 is
  stated in Lemma~\ref{lem:quotient-properties}(1), where $M$ has a
  single neighbor $v$ in $G$.  On the one hand, together with $v$, any
  association set of $G - N_G[M]$ makes an association set of $G$.
  Hence,
  $$\opt{asso}{G} \le \opt{asso}{G - N_G[M]} + 1.$$  On the other hand,
  suppose there is a minimum association set $X$ of $G$ that does not
  contain $v$, then $X\cap M$ is nonempty (recalling that $M$ is not a
  clique). Since $X\setminus M$ is an association set of $G - M$, we
  have $$\opt{asso}{G} = |X| \ge \opt{asso}{G - M} + 1 \ge
  \opt{asso}{G - N_G[M]} + 1.$$ Therefore, the equality must hold,
  which justifies step~6.

  After the algorithm has passed step~6, (1) the graph $G$ is not a
  clique but every maximal strong module of it is a clique; (2) the
  weighted graph $Q$ defined in step~7, which is isomorphic to
  $\widetilde Q(G)$, is $C_3$-free.  Therefore, the output of the
  algorithm of \cite{tu-11-k-path-vertex-cover} is both a dissociation
  set and an association set of $Q$, and its total weight is at most
  $2 \opt{diss}{Q} = 2 \opt{asso}{Q}$.  It is easy to verify that for
  each clique module of $G$, a minimum association set contains either
  all or none of its vertices.  Therefore, from the association set of
  $Q$ we can retrieve an association set of $G$ with the same size,
  which is at most $2 \opt{asso}{Q} = 2 \opt{asso}{G}$.  This
  concludes the proof.
\end{proof}

\section{An efficient implementation}
We now discuss the implementation issues that lead to the claimed
running time.  A simpleminded implementation of the algorithm given in
Fig.~\ref{fig:approximation-outline} takes $O(n^6)$ time, which is
decided by the disposal of induced subgraphs in $\cal F$ (step~3).  It
is unclear to us how to detect them in a more efficient way than the
$O(n^5)$-time enumeration.  But we observe that what we need are no
more than the conditions stipulated in
Lemma~\ref{lem:quotient-properties}, for which being $\cal F$-free is
sufficient but not necessary.  The following relaxation is sufficient
for our algorithmic purpose: We either detect an induced subgraph in
$\cal F$ or determine that $G$ has already satisfied these conditions.
Once a subgraph is found, we can delete all its vertices and repeat
the process.  In summary, we are after an $O(m n)$-time procedure that
finds a set of subgraphs in $\cal F$ such that its deletion leaves a
graph satisfying the conditions of
Lemma~\ref{lem:quotient-properties}.

Toward this end a particular obstacle is the $C_3$-free condition in
the second case of Lemma~\ref{lem:quotient-properties}.  Indeed, the
detection of triangles in linear time is a notorious open problem that
we are not able to solve.  Therefore, we may have to abandon the
simple ``search and remove'' approach.  The first idea here is that we
may dispose of \emph{all} triangles of $\widetilde Q(G)$ in $O(m n)$
time.  This is, however, still not sufficient, because after deleting
a set $X$ of some vertices, its maximal strong modules change, and
more importantly, $\widetilde Q(G - X)$ may \emph{not} be an (induced)
subgraph of $\widetilde Q(G)$.  Our observation is that $\widetilde
Q(G - X)$ is either a clique, an independent set, or an induced
subgraph of $\widetilde Q(G[M])$ for some (not necessarily maximal)
strong module $M$ of $G$.

We start from recalling some simple facts about modular decomposition.
For each maximal strong module $M$ of $G$, we can further take the
maximal strong modules and the quotient graph $\widetilde Q(G[M])$.
This process can be recursively applied until every module consists of
a single vertex.  If we represent each module used in this process as
a node, and add edges connecting every $M$ with all maximal strong
modules of $G[M]$, we obtain a tree rooted at $V(G)$, called the
\emph{modular decomposition tree} of $G$.  The nodes of the {modular
  decomposition tree} are precisely all strong modules of $G$, where
the leaves are all singleton vertex sets, and for every non-leaf node
$M$, its children are the maximal strong modules of $G[M]$
\cite{gallai-67-transitive-orientation}.  It is known that the
{modular decomposition tree} can be constructed in linear time
\cite{mcconnell-99-modular-decomposition-and-transitive-orientation}.
\begin{proposition}\label{lem:prime-quotient}
  If $\widetilde Q(G)$ is prime, then every nontrivial quotient graph
  of $G$ contains $\widetilde Q(G)$ as an induced subgraph.
\end{proposition}
\begin{proof}
  We argue first that the only module that contains a maximal strong
  module of $G$ is $V(G)$.  Suppose for contradiction that there are a
  maximal strong module $M_1$ and another module $M'$ such that
  $M_1\subset M'\subset V(G)$.  By definition, any maximal strong
  module of $G$ intersecting $M'$ must be a proper subset of $M'$.
  Therefore, the module $M'$ contains at least two but not all the
  maximal strong modules of $G$.  However, then the vertices of
  $\widetilde Q(G)$ corresponding to those modules make a nontrivial
  module of the graph $\widetilde Q(G)$, contradicting that it is
  prime.  Now let $\cal P$ be any nontrivial module partition of $G$
  and let $Q$ be the quotient graph defined by it.  Every module in
  $\cal P$ must be a subset of a maximal strong module, and thus by
  taking a module in $\cal P$ for each maximal strong module, we end
  with $\widetilde Q(G)$ as an induced subgraph of $Q$.
\end{proof}

On the one hand, since $V(G)$ itself is a strong module of $G$, every
vertex set $U\subseteq V(G)$ is contained in some strong module.  On
the other hand, since two strong modules are either disjoint or one
containing the other, there is a unique one that is inclusion-wise
minimal of all strong modules containing $U$.

\begin{theorem}\label{thm:subgraph-quotient}
  Let $U\subseteq V(G)$ be a subset of vertices of $G$, and let $M$ be
  the inclusion-wise minimal strong module of $G$ that contains $U$.
  If $\widetilde Q(G[U])$ is prime, then it is a subgraph of
  $\widetilde Q(G[M])$ induced by those maximal strong modules of
  $G[M]$ that intersects $U$.
\end{theorem}
\begin{proof}
  Let $\{M_1, M_2, \ldots, M_p\}$ be the maximal strong modules of the
  subgraph $G[M]$ that intersect $U$.  By definition, for every $1\le
  i\le p$, the set $U\cap M_i$ is a nonempty module of $G[U]$.
  Therefore, $\{U\cap M_1, U\cap M_2, \ldots, U\cap M_p\}$ gives a
  module partition of $U$.  By Proposition~\ref{lem:prime-quotient},
  $\widetilde Q(G[U])$ is an induced subgraph of the quotient graph of
  $G[U]$ defined by this partition.
\end{proof}

We remark that if $\widetilde Q(G[U])$ is a clique or independent set,
then it is not necessarily an induced subgraph of $\widetilde
Q(G[M])$; see, e.g., Fig.~\ref{fig:subgraphs}.

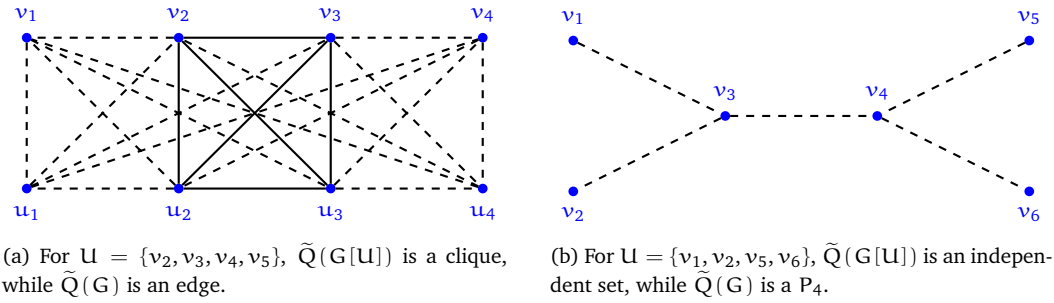
\begin{figure}[h]
  \centering      \footnotesize

  \subfloat[For $U = \{v_2, v_3, v_4, v_5\}$, $\widetilde Q(G[U{]})$
  is a clique, while $\widetilde Q(G)$ is an edge.]{\label{fig:clique}
    \begin{tikzpicture}[auto=left,every node/.style={fill=blue,circle,blue,draw,inner sep=1pt}, every path/.style={thick},scale=1]
      \node[label=above:$v_1$] (v1) at (-3,1) {};
      \node[label=above:$v_2$] (v2) at (-1,1) {};
      \node[label=above:$v_3$] (v3) at (1,1) {};
      \node[label=above:$v_4$] (v4) at (3,1) {};
      \node[label=below:$u_1$] (u1) at (-3,-1) {};
      \node[label=below:$u_2$] (u2) at (-1,-1) {};
      \node[label=below:$u_3$] (u3) at (1,-1) {};
      \node[label=below:$u_4$] (u4) at (3,-1) {};
      \draw[dashed] (u1) -- (v3) -- (v4) -- (u2) -- (u1) -- (v2) -- (v1) -- (u3) -- (u4) -- (v2);
      \draw[dashed] (u1) -- (v4) -- (u4) -- (v1) -- (u1) (u2) -- (v1) (u3) edge (v4) (v3) edge (u4);
      \draw (v2) -- (u2) -- (u3) -- (v3) -- (v2) -- (u3) (v3) edge (u2);
    \end{tikzpicture}    
  }
  \qquad
  \subfloat[For $U = \{v_1, v_2, v_5, v_6\}$, $\widetilde Q(G[U{]})$ is
  an independent set, while $\widetilde Q(G)$ is a
  $P_4$.]{\label{fig:edgeless}
    \begin{tikzpicture}[auto=left,every node/.style={fill=blue,circle,blue,draw,inner sep=1pt}, every path/.style={thick},scale=1]
      \node[label=above:$v_1$] (v1) at (-3,1) {};
      \node[label=below:$v_2$] (u1) at (-3,-1) {};
      \node[label=above:$v_3$] (v2) at (-1,0) {};
      \node[label=above:$v_4$] (v3) at (1,0) {};
      \node[label=above:$v_5$] (v4) at (3,1) {};
      \node[label=below:$v_6$] (u4) at (3,-1) {};
      \draw[dashed] (v1) -- (v2) -- (u1) (v4) -- (v3) -- (u4) (v2) -- (v3);
    \end{tikzpicture}    
    }
    \caption{$\widetilde Q(G[U])$ may not be an induced subgraph of
      $\widetilde Q(G)$.}
  \label{fig:subgraphs}
\end{figure}

\begin{figure}[h!]
  \vspace*{-5mm}
  \setbox4=\vbox{\hsize28pc \noindent\strut
  \begin{quote}
  \vspace*{-5mm} \small

  {\bf Procedure \sc reduce}($G$)
  \\
  {\sc Input}: a graph $G$.
  \\
  {\sc Output}: a set of vertices specified in Lemma~\ref{lem:phase-1}.

  \begin{tabbing}
    AA\=Aa\=AAA\=AAA\=MMMMMMMMMMAAAAAAAAAAAAAA\=A \kill
    0.\> $X \leftarrow \emptyset$; {\bf if} $G$ is a clique {\bf then 
      return} $\emptyset$;
    \\
    1.\> build the modular decomposition tree of $G$;
    \\
    2.\> find all triangles $\cal C$ of $\widetilde Q(G[M])$ for all strong 
    modules of $G$;
    \\
    3.\> {\bf while} $\cal C$ is nonempty {\bf do}
    \\
    3.0.\>\> take a triangle $\{M_1, M_2, M_3\}$ from $\cal C$;
    \\
    3.1.\>\> let $M$  be the  parent of 
    $\{M_1, M_2, M_3\}$ and let $Q = \widetilde Q(G[M])$;
    \\
    3.2.\>\> {\bf if} there is $1\le i\le 3$ such that 
    $M_i \subseteq X$ {\bf then}
    \\
    \>\>\> delete $\{M_1, M_2, M_3\}$ from $\cal C$;
    {\bf continue};
    \\
    3.3.\>\> {\bf if} there are $1\le i < j\le 3$ such that
    $M_i,M_j$ were merged {\bf then}
    \\
    \>\>\> delete $\{M_1, M_2, M_3\}$ from $\cal C$;
    {\bf continue};
    \\
    3.4.\>\> {\bf if} there are $1\le i < j\le 3$ such that 
    $N_Q[M_i] =  N_Q[M_j]$ {\bf then} 
    \\
    \>\>\> merge $M_i$ and $M_j$ into a single module; 
    \\
    \>\>\> delete $\{M_1, M_2, M_3\}$ from $\cal C$;
    {\bf continue};
    \\
    3.5.\>\> find $M'\in V(Q)$ that is adjacent to only one of $M_1$ and $M_2$; 
    \\
    \>\> $\backslash\!\!\backslash$
    {\em Assume that $M'$ is adjacent to $M_2$ but not $M_1$.}
    \\
    3.6.\>\> {\bf if} $M'$ is adjacent to $M_3$ {\bf then} 
    \\
    \>\>\> find $M''\in V(Q)$ adjacent to only one of $M_2$ and $M_3$; 
    \\
    \>\>\> find  gems or darts or $C_4$ and move their vertices to $X$; 
    \\
    3.7.\>\> {\bf else} \comment{$M'$ is nonadjacent to $M_3$.}
    \\
    \>\>\> find $M''\in V(Q)$ adjacent to only one of $M_1$ and $M_3$; 
    \\
    \>\>\> find darts or bulls or $C_4$ or gem and move their vertices to $X$; 
    \\
    4.\> {\bf repeat} until $X$ is unchanged {\bf do} 
    \\
    4.0. \>\>  $G \leftarrow G - X$; 
    \comment{ All modules below are maximal strong modules of $G$.}
    \\
    4.1. \>\> {\bf if} $G$ is not connected {\bf then} 
    \\
    \>\>\>{\bf return} 
    $\bigcup$\{{\sc reduce}($C$): $C$ is a component of $G$\}$\cup X$;
    \\
    4.2. \>\> {\bf if} $G$ is a clique or consists of two intersecting cliques 
    {\bf then return} $X$;
    \\
    4.3.\>\>let $Q$ denote $\widetilde Q(G)$;
    \comment{$Q$ is a clique or $C_3$-free.}
    \\
    4.4.\>\> {\bf if} two non-clique  modules are adjacent in $Q$ {\bf then}
    \\
    \>\>\> find a $C_4$ of $G$ and move its vertices to $X$;
    \\
    4.5.\>\> {\bf else if} $Q$ is a clique but not an edge {\bf then} 
    \\
    4.5.0.\>\>\> let $M$ be the only nontrivial module;
    \comment{It exists because $G$ is not a clique.}
    \\
    4.5.1.\>\>\> {\bf if} $G[M]$ has a $P_4$ or $C_4$ {\bf then} 
    \\
    \>\>\>\> find a gem or $C_4$ of $G$ and move its vertices to $X$;
    \\
    4.5.2.\>\>\> {\bf else if} $G[M]$ has two components {\bf then} 
    \\
    \>\>\>\> find a dart of $G$ and move its vertices to $X$;
    \\
    4.5.3.\>\>\> {\bf else} find a fox of $G$ and move its vertices to $X$;
    \\
    4.6.\>\> {\bf else if} a module $M$ is not a cluster {\bf then} 
    \\
    4.6.1.\>\>\> {\bf if} $G[M]$ has a $P_4$ or $C_4$ {\bf then} 
    \\
    \>\>\>\> find a gem or $C_4$ of $G$ and move its vertices to $X$;
    \\
    4.6.2.\>\>\> {\bf else if} $G[M]$ is not connected {\bf then} 
    \\    
    \>\>\>\> find a dart of $G$ and move its vertices to $X$;
    \\
    4.6.3.\>\>\> {\bf else} find a dart of $G$ and add its vertices to $X$;
    \comment{$Q$ is not a clique.}
    \\
    4.7.\>\> {\bf else if} a non-clique module $M$ has two neighbors $M_1,M_2$ 
    {\bf then} 
    \\
    \>\>\> find a $C_4$ of $G$ and move its vertices to $X$;
    \\
    4.8.\>\> {\bf else if} a nontrivial module $M$ has neighbors $M_1$, $M_2$, 
    $M_3$ {\bf then} 
    \\
    \>\>\> find a fox of $G$ and move its vertices to $X$;
    \\
    5.\> {\bf return} $X$.
  \end{tabbing}
\end{quote} \vspace*{-3mm} \strut} $$\boxit{\box4}$$
\vspace*{-7mm}
\caption{Procedure for the first phase.}
\label{fig:alg-small-subgraph}
\end{figure}

We are now ready to present the efficient implementation for the first
phase, which would replace the first three steps of algorithm {\sc
  approx-asso} (Fig.~\ref{fig:approximation-outline}).  
\begin{lemma}\label{lem:phase-1}
  In $O(m n)$ time we can find a set $\cal H$ of disjoint induced
  subgraphs of $G$ such that each $H\in \cal H$ is in $\cal F$ and $G
  - \bigcup_{H\in \cal H} V(H)$ satisfies the conditions of
  Lemma~\ref{lem:quotient-properties}.
\end{lemma}
\begin{proof}
  We use the procedure described in Fig.~\ref{fig:alg-small-subgraph}.
  Step~0 is trivial.  Step~1 uses the algorithm of McConnell and
  Spinrad~\cite{mcconnell-99-modular-decomposition-and-transitive-orientation},
  and step~2 uses simple enumeration.  This leads us to the disposal
  of triangles in step~3.  During its progress, a maximal strong
  module $M$ of the input graph $G$ may not remain a maximal strong
  module of the current graph (i.e., $G - X$).  But if $M$ is not
  completely deleted (i.e., $M\not\subseteq X$), then its remnant
  (i.e., $M\setminus X$) is always a module of $G - X$.

  Note that the three modules in each triangle must have the same
  parent in the modular decomposition tree.  For each triangle $\{M_1,
  M_2, M_3\}$, we focus on their parent $M$ (in the modular
  decomposition tree) and the subgraph $G[M]$ (step~3.1).  All the
  modules mentioned in steps 3.2--3.7 are maximal strong modules of
  subgraph $G[M]$; they correspond to $V(Q)$.

  If either of the conditions of steps~3.2 and~3.3 is true, then the
  triangle has been disposed of and we can continue to the next one.
  If the deletion of vertices in previous iterations has made
  $N_Q[M_i] = N_Q[M_j]$ for some $1\le i < j\le 3$, then $M_i\cup M_j$
  is a module of $G[M\setminus X]$.  Note that after they are merged,
  both $M_i$ and $M_j$ refer to the new module.  Now that the
  procedure has passed steps~3.2--3.4, for each $1\le i < j\le 3$, we
  can find a module adjacent to only one of $M_i$ and $M_j$.  This
  justifies step~3.5, and we may assume without loss of generality
  that the module $M'$ is adjacent to $M_2$ but not $M_1$; the other
  case can be dealt with a symmetric way, which is omitted.  In
  step~3.6, depending on the adjacency between $M''$ and $M_1, M'$, we
  are in one of the following three cases:
  \begin{inparaitem}
  \item if $M''$ is adjacent to  neither of $M_1, M'$, then there is a dart;
  \item if $M''$ is adjacent to precisely one of $M_1, M'$, then there
    is a gem; or
  \item otherwise ($M''$ is adjacent to both of $M_1, M'$), there is a
    $C_4$;
  \end{inparaitem}
  This forbidden subgraph can be constructed by taking one vertex from
  each of $M_1, M_2, M_3, M'$, and $M''$.  We can actually find $\min
  \{|M_1|, |M_2|, |M_3|, |M'|, |M''|\}$ number of gems or darts, or
  $\min \{|M_1|, |M_3|, |M'|, |M''|\}$ number of $C_4$'s, which we all
  move into $X$.  It is similar for step~3.7.

  After step~3, $G$ might become disconnected.  Then we work on its
  components one by one.  Steps~4.1 and~4.2 are simple.  The fact that
  the quotient graph $Q$ built in step~4.3 is either a clique or is
  $C_3$-free can be argued using Theorem~\ref{thm:subgraph-quotient}.
  Suppose for contradiction that $Q$ is prime but contains a $C_3$.
  Then by Theorem~\ref{thm:subgraph-quotient}, it is a subgraph of
  $\widetilde Q(G[M])$ for some strong module $M$ of $G$.  Let $\{M_1,
  M_2, M_3\}$ be the triangle of $\widetilde Q(G[M])$ corresponding a
  triangle in $Q$.  But in step~3, either one of $\{M_1, M_2, M_3\}$
  has been completely put into $X$, or two of them have been merged
  (then unless $Q$ is a clique or an independent set, they will always
  be in the same maximal strong module).  Therefore, $Q$ must be
  $C_3$-free if it is not a clique.

  Note that the algorithm enters at most one of steps~4.4--4.8.  The
  correctness of step~4.4 is clear.  If $Q$ is a clique and it passes
  step~4.4, then all but one maximal strong module are trivial.  (Note
  that each universal vertex is a maximal strong module.)  A $P_4$ of
  $M$ together with a vertex in $N_G(M)$ makes a gem (4.5.1).  Now
  that $G[M]$ is $\{P_4, C_4\}$-free, and has no universal vertex (a
  universal vertex of $G[M]$ is a universal vertex of $G$ as well),
  according to Proposition~\ref{lem:trivially-perfect}, $G[M]$ is
  disconnected.  Since step~4.2 does not apply, in step~4.5.2, at
  least one component is not a clique, and has a $P_3$, which,
  together with a vertex $u\in N_G(M)$ and any vertex from another
  component of $G[M]$, makes a dart.  Otherwise (step~4.5.3), $G[M]$
  has at least three components, and we can find a fox by taking three
  vertices from different components of $G[M]$ and two vertices from
  $N_G(M)$ (recall that when it enters step~4.5, $G$ must have at
  least two universal vertices).  Step~4.6 can be argued similar as
  step~4.5, where one should note that if there are only two modules
  then $G[M]$ cannot have universal vertices.  After them, $Q$ is
  always $C_3$-free.  Therefore, modules $M_1, M_2$ in step~4.7 and
  modules $M_1, M_2, M_3$ in step~4.8 are (pairwise) nonadjacent.

  If $G$ consists of two intersecting cliques, the procedure returns
  at step~4.2.  Hence we may assume that it is not the case.  The
  quotient graph $\widetilde Q(G)$ is $C_3$-free because
  Theorem~\ref{thm:subgraph-quotient} and the algorithm has passed
  step~4.5.  Conditions (1) and (2) of
  Lemma~\ref{lem:quotient-properties} follow from the correctness
  argument for steps~4.6--4.8.  We now calculate the running time of
  the procedure.  Note that the total number of edges of the subgraphs
  induced the strong modules of $G$ is upper bounded by $m$.  Thus,
  all the triangles can be listed in $O(m n)$ time in step~2.  Each
  iteration of step~3 takes $O(m)$ time, and it decreases the order of
  $Q$ by at least one, and thus step~3 takes $O(m n)$ time in total.
  Each iteration of step~4 takes $O(m)$ time, and it decreases the
  order of $G$ by at least one, and hence step~4 takes $O(m n)$ time
  in total.  This concludes the proof of this lemma.
\end{proof}

{
  \small 

}
\end{document}